\documentclass[11pt]{article}
\usepackage{amsmath, amssymb, amsthm}
\usepackage{graphicx}

\usepackage{mathrsfs}

\renewcommand{\baselinestretch}{1.25}

 \newtheorem{thm}{Theorem}[section]
 \newtheorem{cor}[thm]{Corollary}

 \newtheorem{obs}[thm]{Observation}
 \theoremstyle{definition}
 \newtheorem{defn}[thm]{Definition}
 \newtheorem{exmp}[thm]{Example}

 \theoremstyle{remark}

\numberwithin{equation}{section}

 \newcommand{\address}[1]
 { \vspace{-2em}\begin{center}
  \footnotesize{#1}
   \end{center}}

  \newcommand{\email}[1]
   {\vspace{-2em} \begin{center}
   \footnotesize{{\it E-mail address:} \texttt{#1}}
   \end{center}}

 \newcommand{\Real}{\mathbb{R}}

 \newcommand{\set}[1]{\left\{#1\right\}}
 \newcommand{\Set}[2]{\set{#1\ \vert\ #2}}

\newcommand{\BOX}[1]{\texttt{box}\,(#1)}

\title{\Large\bf Ferrers Dimension and Boxicity}

\author{Soumyottam Chatterjee$^{*}$ and Shamik Ghosh$^{+}$}

\begin{document}

%%% ---------------------------------------------------------------
\maketitle

\vspace{-0.5em}
\address{$^{*}$Department of Electronics and Tele-Communications Engineering and $^{+}$Department of Mathematics,\\
Jadavpur University,
Kolkata - 700 032, India.}

\vspace{0.25em}
\email{$^{*}$soumyottamchatterjee@gmail.com, $^{+}$sghosh@math.jdvu.ac.in}

%%%%%%%%%%%%%%%%%%%%%%%%%%%%%%%%%%%%%%%%%%%%%%%%%%%%%%%%%%%%%%%%%%

\renewcommand{\baselinestretch}{1}
\begin{abstract}
This note explores the relation between the boxicity of undirected graphs and the Ferrers dimension of digraphs.

\vspace{0.5em}\noindent
\end{abstract}

%%% ------------------------------------------------------------------

\renewcommand{\baselinestretch}{1.25}

%%% ------------------------------------------------------------------

\section{Introduction}

An undirected graph $G = (V, E)$ is an {\em interval graph} if and only if it is the intersection graph of a family of intervals on the real line. Each vertex is assigned an interval and two vertices are adjacent if and only if their corresponding intervals intersect. Motivated by theoretical as well as practical considerations, graph theorists have tried to generalize the concept of interval graphs in many ways. In many cases, representation of a graph as the intersection graph of a family of geometric objects, which are generalizations of intervals is sought. An example is the concept of boxicity introduced by F. S. Roberts in 1969 \cite{F}. For a graph $G$, its {\em boxicity} $\BOX{G}$ is the minimum positive integer $b$ such that $G$ can be represented as the intersection graph of axis-parallel b-dimensional boxes. Here a b-dimensional box is a Cartesian product $I_1\times I_2\times\cdots\times I_b$ where each $I_i$ is a closed interval on the real line. The boxicity of a complete graph may be assumed to be zero and since a one-dimensional box is a closed interval on the real line, graphs of boxicity at most 1 are exactly the interval graphs.

\vspace{1 em} Introduced independently by Guttman \cite{G} and Riguet \cite{R}, a {\em Ferrers digraph} $D = (V, E)$ is a directed graph (in short, digraph) whose successor sets are linearly ordered by inclusion, where the successor set of $v \in V$ is its set of out-neighbors $\Set{u \in V}{vu \in E}$. It is easy to see that the successor sets are linearly ordered by inclusion if and only if the analogously defined predecessor sets are linearly ordered by inclusion, and that both are equivalent to the transformability of the adjacency matrix by independent row and column permutations to a $(0,1)$-matrix in which the 1's are clustered in a corner in the shape of a Ferrers diagram (hence the term `Ferrers digraph'). It is well-known that every digraph $D$ is the intersection of a finite number of Ferrers digraphs and the minimum such number is its {\em Ferrers dimension}. It is known \cite{R} that a digraph D is a Ferrers digraph if and only if its adjacency matrix does not contain any $2\times 2$ permutation matrix:
$$\left(
\begin{array}{cc}
1 & 0\\
0 & 1
\end{array}
\right)
\qquad \textrm{ or }\qquad
\left(
\begin{array}{cc}
0 & 1\\
1 & 0
\end{array}
\right).$$
The digraphs of Ferrers dimension at most 2 were characterized by Cogis \cite{C}. He called every $2\times 2$ permutation matrix a {\em couple} and defined an undirected graph $H(D)$, the graph {\em associated to a digraph} $D$ whose vertices correspond to the 0's of its adjacency matrix with two such vertices joined by an edge if and only if the corresponding 0's belong to a couple. Cogis \cite{C} proved that $D$ is of Ferrers dimension at most 2 if and only if $H(D)$ is bipartite. In the general case, if $d_F(D) = n$, then there exist Ferrers digraphs $F_i$, $i = 1, 2, \ldots , n$, such that $D$ can be expressed as $D = F_1 \cap F_2 \cap \cdots \cap F_n$. Zeros belonging to any particular $F_i$ do not form any couple among themselves and consequently form an independent set in $H(D)$. Thus $\chi(H(D)) \leq d_F(D)$ where $\chi(H(D))$ is the chromatic number of $H(D)$. No instance has been found yet for which the inequality is a strict one and it is not known whether $\chi(H(D)) = d_F(D)$ for all digraphs $D$. But from the above inequality, it follows that $d_F(D)\geqslant n$ whenever $H(D)$ contains $K_n$. In fact, $d_F(D)=n$ if $H(D)=K_n$, as $d_F(D)$ cannot exceed the number of $0$'s of the adjacency matrix of $D$.

\vspace{1em} Let $G=(V,E)$ be a graph (directed or undirected). We denote the adjacency matrix of $G$ by $A(G)$. For convenience, an entry of $A(G)$ corresponding to, say, the vertex $u_i\in V$ in the row and the vertex $v_j\in V$ in the column will be denoted by, simply, $u_iv_j$. The graph whose adjacency matrix is obtained by interchanging $0$'s and $1$'s of $A(G)$ will be denoted by $\overline{G}$. Note that if $G$ has loops at all vertices (i.e., all principal diagonal elements of $A(G)$ are $1$), then $\overline{G}$ is a graph without loops (i.e., all principal diagonal elements of $A(\overline{G})$ are $0$) and vice-versa. Again for a digraph $D$ with adjacency matrix $A(D)$, we denote the digraph whose adjacency matrix is $A(D)^T$ (the transpose of the matrix $A(D)$) by $D^T$.

\vspace{1em} Now we explore some nice relations between Ferrers digraphs and interval graphs. A digraph $D=(V,E)$ is {\em oriented} if every arc of $D$ has a unique direction (i.e., $uv\in E\Longrightarrow vu\notin E$ for any $u,v\in V$). An oriented digraph $D=(V,E)$ is {\em transitively oriented} if $a,b,c\in V,\ ab,bc\in E\Longrightarrow ac\in E$. An undirected graph $G=(V,E)$ is {\em transitively orientable} if each edge of $G$ can be assigned a one-way direction in such a way that the resulting digraph is transitively oriented. We call this digraph as a {\em transitive orientation} of $G$. A transitive digraph $D=(V,E)$ without loop at any vertex is an {\em interval order} digraph if $a,b,x,y\in V,\ ax, by\in E\Longrightarrow ay \textrm{ or } bx \in E$. The class of interval order digraphs are transitive digraphs $D$ such that $\overline{D}\cap\overline{D}^T$ are interval graphs \cite{Fi} or equivalently, loopless Ferrers digraphs \cite{P}. Also if $F$ is a Ferrers digraph without loops, then $\overline{F}$ is a Ferrers digraph with loop at every vertex. Further we know that an undirected graph is an interval graph if and only if it does not contain $C_4$ (the cycle of length $4$) as an induced subgraph and its complement (called {\em co-interval} graph) is transitively orientable.\cite{A} Finally since every orientation of $\overline{C_4}=2K_2$ is isomorphic to $D_1$ (cf. Figure \ref{fig:d1}), we have the following observations:

\begin{obs}\label{obs:alpha}
Let $I$ be an undirected graph (with loop at every vertex). Then $I$ is an interval graph if and only if there exists a Ferrers digraph $F$ (with loop at every vertex) such that $I=F\cap F^T$.
\end{obs}

\begin{obs}\label{lem:beta}
A digraph without loop at any vertex is a Ferrers digraph if and only if it is transitively oriented and does not contain $D_1$ as an induced subdigraph.
\end{obs}

\begin{figure}[h]
\begin{center}
\includegraphics*[scale=0.5]{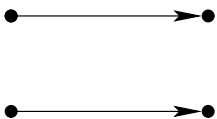}
\caption{The digraph $D_1$}\label{fig:d1}
\end{center}
\end{figure}

Note that $D_1$ itself is a transitively oriented digraph without loops  and the following is an example of an oriented digraph (without loops) which does not contain $D_1$ as an induced subdigraph, but it is not a Ferrers digraph as it is not transitively oriented:

\begin{figure}[h]
\begin{center}
\includegraphics*[scale=0.5]{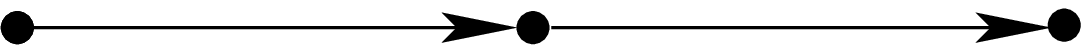}
%\caption{The digraph $D_1$}\label{fig:d1}
\end{center}
\end{figure}

\noindent Moreover the following digraph is transitively oriented and does not contain $D_1$ as an induced subdigraph, though it is not a Ferrers digraph.

\begin{figure}[h]
\begin{center}
\includegraphics*[scale=0.5]{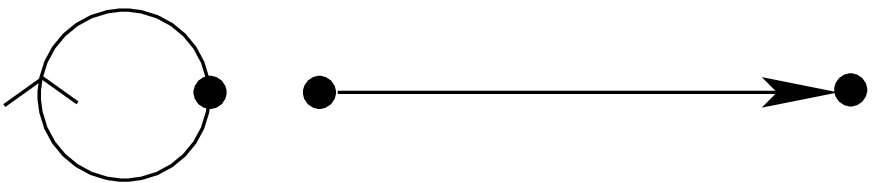}
\end{center}
\end{figure}

The following are some interesting consequences of the above observations:

\begin{cor}\label{cor:chi}
Every transitive orientation of a co-interval graph (without loops) is a Ferrers digraph.
\end{cor}

\begin{cor}\label{cor:delta}
An undirected graph $I$ (with loop at every vertex) is an interval graph if and only if $\overline{I}$ has an orientation of a Ferrers digraph (without loops). 
\end{cor}

The {\em intersection digraph} $D=(V,E)$ of a family of ordered pairs of sets $\Set{(S_u,T_u)}{u\in V}$ is the digraph such that $uv\in E$ if and only if $S_u\cap T_v\neq\emptyset$. An {\em Interval digraph} is an intersection digraph of a family of ordered pairs of intervals on the real line. A bipartite graph (in short, {\em bigraph}) $B(X,Y,E)$ is an {\em intersection bigraph} if there exist a family $\mathcal{F} = \{ I_v : v \in X\cup Y\}$ of sets such that $uv\in E$ ($u\in X,\ v\in Y$) if and only if $I_u\cap I_v \neq \emptyset$. An {\em interval bigraph} is such when each $I_v$ is an interval on the real line. The submatrix of the adjacency matrix of $B$ consisting of the rows corresponding to one partite set and the columns corresponding to the other is known as the {\em biadjacency matrix} of $B$. It should be noted that the two concepts intersection digraph and  intersection bigraph are basically equivalent \cite{P}. Indeed the bigraph whose biadjacency matrix is the adjacency matrix of a digraph corresponds to the digraph. Also given any bigraph, if the number of vertices of the sets $X$ and $Y$ are not equal, we can make them equal by properly introducing some isolated vertices (correspondingly adding the required number of rows and columns consisting of all zeros in the biadjacency matrix of the bigraph) and then convert it into the adjacency matrix of a digraph. The bigraph corresponding to a Ferrers digraph is known as {\em Ferrers bigraph} and the {\em Ferrers dimension} of a bigraph $B$ is the minimum number of Ferrers bigraphs whose intersection is $B$. Now we shall observe an interesting relation between Ferrers bigraphs and interval graphs.

\begin{defn} 
Let $B$ be a bigraph with biadjacency matrix $A$. Then the graph with the following adjacency matrix is denoted by $\widehat{B}$:
$$\begin{array}{|c|c|}
\hline 
\mathbf{1} & A \\
\hline 
A^T & \mathbf{1} \\
\hline 
\end{array}$$
Clearly the graph $\widehat{B}$ is obtained from the bigraph $B$ by joining edges so that the partite sets of $B$ become cliques and by adding loops at all vertices.
\end{defn}

Let $M$ be a symmetric $(0,1)$ matrix with $1$'s in the principal diagonal. Then $M$ is said to satisfy the {\em quasi-linear property for ones} if $1$'s are consecutive right to and below the principal diagonal. It is known \cite{M} that an undirected graph $G$ (with loop at every vertex) is an interval graph if and only if rows and columns of $A(G)$ can be suitably permuted (using the same permutation for rows and columns) in such a way that it satisfies the quasi-linear property for ones. Now if $F$ is a Ferrers bigraph, then it is interesting to note that $\widehat{F}$ is an interval graph (with loop at every vertex), as its adjacency matrix has quasi-linear property for ones.

\begin{figure}[h]
\begin{center}
\includegraphics*[scale=0.25]{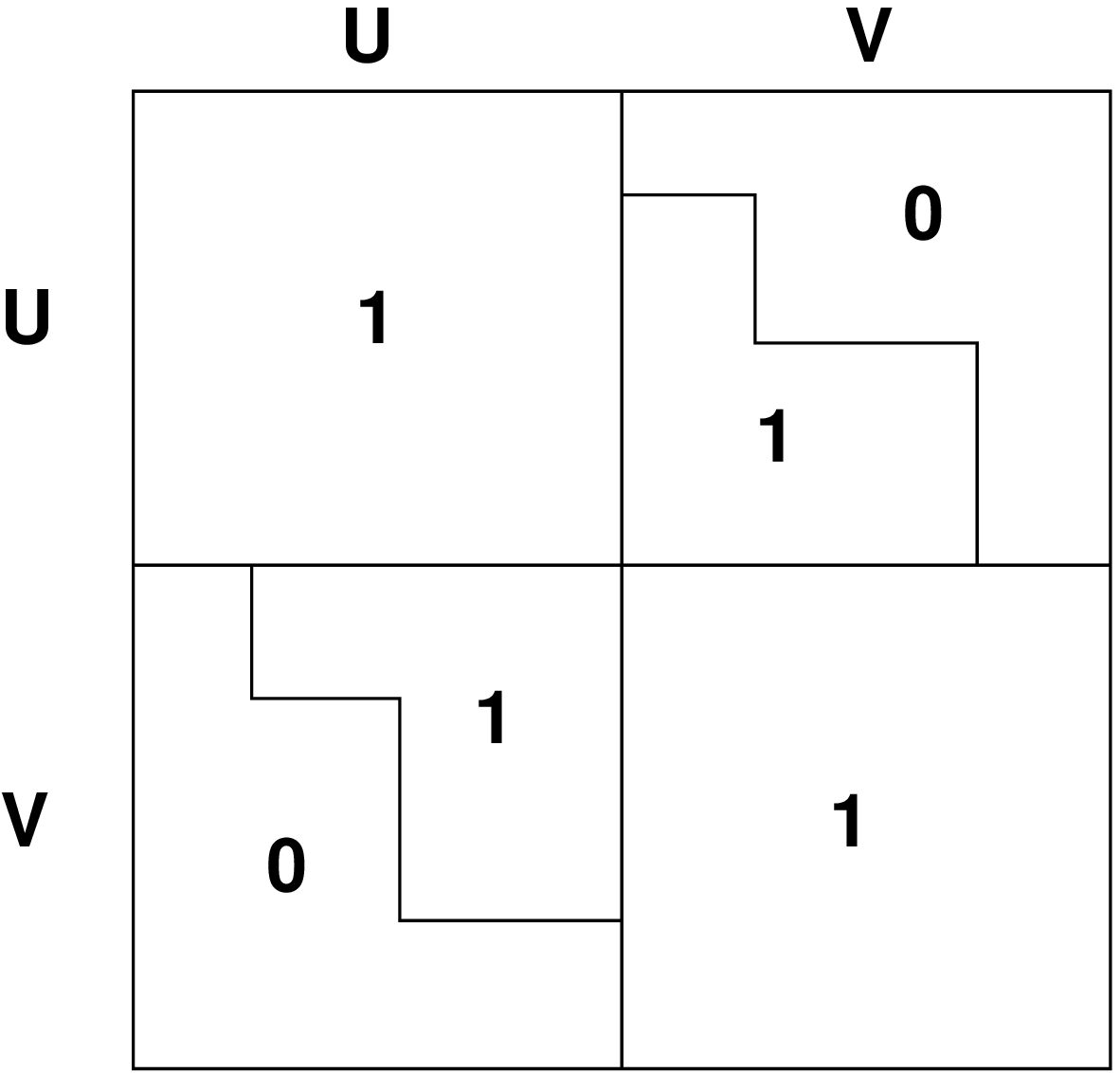}
\end{center}
\end{figure}

Conversely, consider an interval graph $I$ whose vertices are covered by two disjoint cliques, say $X$ and $Y$. We call such an interval graph, a {\em $2$-clique interval graph}. Now since $I$ is an interval graph, its maximal cliques are consecutively ordered. Let $\set{C_1,C_2,\ldots C_r}$ be a consecutive linear ordering of maximal cliques. Assign (closed) intervals $I_v$ to the each vertex $v$ according to its first and last appearance in above sequence of maximal cliques. Let $X\subseteq C_i$ and $Y\subseteq C_j$. 

\vspace{1em} Suppose $i<j$. Now for every $x\in X,\ i\in I_x$ and for all $y\in Y,\ j\in I_y$. We may restrict the right end point of each $I_x$ up to $j$ whenever it is exceeding $j$ as every $I_y$ contains $j$ and each $I_x$ has already a common point, namely $i$, for every other vertex in $X$. Similarly restrict the left end points of $I_y$ up to $i$ whenever it is lower than $i$. 

\vspace{1em} With this new assignment of intervals for the interval graph $I$, we go for further reduction. Now since for every $y\in Y$, the left end point of $I_y$ is $\geqslant i$ and $i\in I_x$ for all $x\in X$, safely we may fix all the left end points of $I_x$ to $i$ and similarly all right end points of $I_y$ to $j$.

\vspace{1em} Finally we arrange all the vertices of $I$ in its adjacency matrix according to the lexicographic ordering (dictionary order) of the above constructed intervals. Thus the adjacency matrix of $I$ (with loop at every vertex) takes the following form:
$$\begin{array}{c|cc}
\multicolumn{1}{c}{} & X & Y \\ \cline{2-3} 
X & \multicolumn{1}{c|}{\mathbf{1}} & \multicolumn{1}{c|}{A} \\ \cline{2-3} 
Y & \multicolumn{1}{c|}{A^T} & \multicolumn{1}{c|}{\mathbf{1}} \\ \cline{2-3}
\end{array}$$
Moreover, in this matrix, $x_i y_j=0$ if and only if $I_{x_i}<I_{y_j}$ which gives us $x_i y_j=0\Longrightarrow x_i y_k=0$ for all $k\geqslant j$ and $x_i y_j=1\Longrightarrow x_r y_j=1$ for all $r\leqslant i$. That is in each row of the submatrix $A$, every $0$ has only $0$ to its right and every $1$ has only $1$ below it. What this says is nothing but the bigraph corresponding to the biadjacency matrix $A$ is a Ferrers bigraph. The case for $i>j$ is similar. In this case the vertices of $Y$ would come before those in $X$ in the adjacency matrix of $I$. Thus we have the following result:

\begin{obs}\label{t:fchar}
A bigraph $B$ is a Ferrers bigraph if and only if $\widehat{B}$ is a ($2$-clique) interval graph.
\end{obs}

It is interesting to note that every $2$-clique interval graph $I$ is necessarily an {\em indifference graph}\footnote{Equivalently, a {\em proper interval graph} (an interval graph with an interval representation where no interval properly contains another) or a {\em unit interval graph} (which has an interval representation with all the intervals are of same length) or an interval graph which does not contain an induced copy of $K_{1,3}$.} as $I$ does not contain an induced $K_{1,3}$ (since among any three vertices of $I$, two of them must be in the same clique). 
Also since the bigraph complement (also called the {\em converse}) of a Ferrers bigraph is again a Ferrers bigraph, the above observation immediately gives the following:

\begin{cor}\label{c:fcomp}
A bigraph $B$ is a Ferrers bigraph if and only if its graph complement is a $2$-clique interval graph (with loop at every vertex).\footnote{The result is analogous to a known one which states that a bigraph $B$ is of Ferrers dimension at most $2$ if and only if its graph complement is a $2$-clique circular-arc graph (with loop at every vertex).}
\end{cor}

In this note, we relate the two concepts - one corresponding to undirected graphs and the other to directed graphs - those of boxicity and Ferrers dimension respectively and propose a new construction for determining the Ferrers dimension of a digraph in the general case.

\section{Relating boxicity with Ferrers dimension}

An application to Observation \ref{obs:alpha} leads to the following theorem. Henceforth we denote the Ferrers dimension of a digraph $D$ [bigraph $B$] by $d_F(D)$ [resp. $d_F(B)$].

\begin{thm} \label{thm:alpha}
Let $G$ be an undirected graph with loop at every vertex. Then there exists a digraph $D$ such that $G=D\cap D^T$ and $\BOX{G}=d_F(D)$. In general, $\BOX{G}\leqslant d_F(D)$ for any digraph $D$ such that $G=D\cap D^T$. Cosequently, 
$$\BOX{G}=\min \Set{d_F (D)}{G=D\cap D^T \textrm{ for some digraph } D}.$$
\end{thm}

\begin{proof}Let $n=\BOX{G}$. Then $G = I_1 \cap I_2 \cap \cdots \cap I_n$, where each $I_i$ is an interval graph with loop at every vertex for $i = 1, 2, \ldots , n$. Also by Observation \ref{obs:alpha}, for each $i = 1, 2, \ldots, n$, $I_i = F_i \cap {F_i}^T$ for some Ferrers digraph $F_i$ (with loop at every vertex). Then $G=D\cap D^T$, where $D = F_1 \cap F_2 \cap \cdots \cap F_n$. As $D$ can be expressed as the intersection of $n$ Ferrers digraphs, $d_F(D) \leq n$. We show that $d_F(D)$ is exactly equal to $n$. If possible, let $d_F(D) = m$ where $m < n$. Then there exist Ferrers digraphs $F^\prime_i$, for $i = 1, 2, \ldots , m$, for which $D = F^\prime_1 \cap F^\prime_2 \cap \cdots \cap F^\prime_m$. Now $G=D \cap D^T={I^\prime_1} \cap {I^\prime_2} \cap \cdots \cap {I^\prime_m}$, where $I^\prime_i = {F^\prime_i} \cap {{F^\prime_i}^T}$ for $i = 1, 2, \ldots , m$. Again since the graph $G$ has loops at all the vertices and $G=D\cap D^T$, the digraph $D$ and hence each $F_i^\prime$ also has loops at all its vertices. Then by Observation \ref{obs:alpha}, each $I^\prime_i$'s is an interval graph. So $G$ can be expressed as the intersection of $m$ interval graphs, where $m < n$, contrary to the fact that the boxicity of $G$ is $n$. Hence $d_F(D) = n$.

\vspace{1em} Moreover from the above deduction, it follows that, whenever $G=D\cap D^T$ for some digraph $D$, we have $\BOX{G}\leqslant d_F(D)$. This completes the proof.
\end{proof}

On the other hand the following is a consequence of Observation \ref{t:fchar}:

\begin{thm}
Let $B$ be a bipartite graph. Then $d_F(B)=\BOX{\widehat{B}}$.
%the Ferrers dimension of $B$ is equal to the boxicity of the graph $$.
\end{thm}

\begin{proof}
Suppose the bigraph $B$ is of Ferrers dimension $m$. Then $B=F_1\cap F_2\cap\cdots \cap F_m$ for some Ferrers bigraphs $F_i$, {\small ($i=1,2,\ldots ,m$)} which implies $\widehat{B}=\widehat{F_1}\cap \widehat{F_2}\cap\cdots \cap \widehat{F_m}$. Since each $\widehat{F_i}$ is an interval graph by Observation \ref{t:fchar}, we have $n\leqslant m$, if the graph $\widehat{B}$ has boxicity $n$. 

\vspace{1em} Conversely, if $n$ is the boxicity of $\widehat{B}$, then $\widehat{B}=I_1\cap I_2\cap\cdots\cap I_n$ where each $I_j$ is an interval graph. Also since their intersection (the graph $\widehat{B}$) has two cliques covering all the vertices, each $I_j$ also contains same cliques for those vertices, i.e., each of them is a $2$-clique interval graphs and the two cliques are consisting of the partite sets of $B$. Thus it follows from Observation \ref{t:fchar} that $B$ is the intersection of $n$ Ferrers bigraphs, $F_1,F_2,\ldots ,F_n$ such that $\widehat{F_j}=I_j$ for all $j=1,2,\ldots ,n$. Therefore $m\leqslant n$, as required.
\end{proof}

As an immediate consequence of the above theorem we obtain certain  characterizations of bigraphs of Ferrers dimension $2$ and interval bigraphs.

\begin{cor}
A bipartite graph $B$ is of Ferrers dimension at most $2$ if and only if $\widehat{B}$ is a $2$-clique rectangular graph.\footnote{A {\em rectangular graph} is an intersection graph of rectangles in $\Real^2$. A {\em $2$-clique rectangular graph} is a rectangular graph whose vertices are covered by two disjoint cliques.}
\end{cor}

\begin{cor}
A bipartite graph $B$ is an interval bigraph if and only if $\widehat{B}$ is a $2$-clique rectangular graph such that there is a rectangular representation of $\widehat{B}$ in which for every pair of rectangles, their projections intersect on at least one of the axes. 
\end{cor}

\begin{proof}
The proof follows from the fact that

\vspace{1em} \begin{tabular}{cl}
 & $B$ is an interval bigraph\\
$\Longleftrightarrow$ & $B=F_1\cap F_2$ where $F_1$ and $F_2$ are two Ferrers bigraphs whose union is complete \cite{M}\\
$\Longleftrightarrow$ & $\widehat{B}=\widehat{F_1}\cap \widehat{F_2}$ for two Ferrers bigraphs, $F_1,F_2$ with $\widehat{F_1}\cup \widehat{F_2}$ is complete\\
$\Longleftrightarrow$ & $\widehat{B}=I_1\cap I_2$ where $I_1$ and $I_2$ are ($2$-clique) interval graphs whose union is complete.
\end{tabular}

\end{proof}

Let $G$ be an undirected graph. Denote the corresponding (symmetric) digraph with the same adjacency matrix as that of $G$ by $D(G)$.

\begin{thm} \label{thm:beta}
Let $G$ be an undirected graph $G$ (with loop at every vertex) such that $\BOX{G}=b$. Let $D(G)$ be the corresponding digraph with the same adjacency matrix as that of $G$ and $k=d_F(D(G))$. Then 
$$\frac{k}{2} \leq b \leq (k-1),$$
and the bounds are tight.
\end{thm}

\begin{proof}
Since $b = \BOX{G}$, $G$ can be expressed as $G = I_1 \cap I_2 \cap \cdots \cap I_b$, where each $I_i$ is an interval graph and so $I_i=F_i\cap F^T_i$ for some Ferrers digraphs (with loop at every vertex) for $i=1,2,\ldots ,b$. Then $D(G)= (F_1 \cap {F_1}^T) \cap (F_2 \cap {F_2}^T) \cap \cdots \cap (F_b \cap {F_b}^T)$ which implies $k=d_F(D(G))\leqslant 2b$, i.e., $\frac{k}{2}\leqslant b$. The limit is reached in the case of $G=C_4$ as $\BOX{C_4}=2$ and from the following adjacency matrix of $D=D(C_4)$ it is clear that $H(D)=K_4$ and hence $d_F(D(C_4))=4$.

$$\begin{array}{cc}
\includegraphics[scale=0.4]{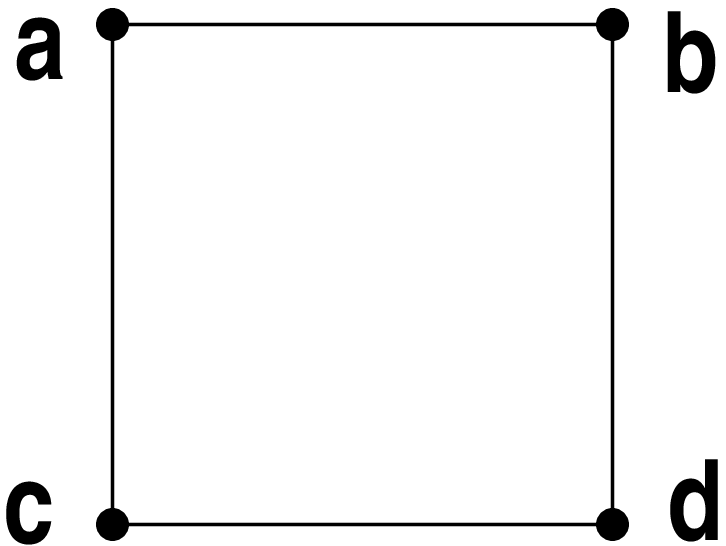} & \hspace{1.4in} 
\begin{array}[b]{c|cccc}
\multicolumn{1}{c}{} & a & b & c & d \\ \cline{2-5}
a & 1 & 1 & 1 & 0 \\
b & 1 & 1 & 0 & 1 \\
c & 1 & 0 & 1 & 1 \\
d & 0 & 1 & 1 & 1 
\end{array}
\end{array}$$

\vspace{1em}As for the upper bound, let $d_F(D(G))=k$. Then $D=D(G) = F_1 \cap F_2 \cap \cdots \cap F_k$. Since $D$ is symmetric, $D=D^T=D\cap D^T=G=(F_1 \cap F_2 \cap \cdots \cap F_k) \cap ({F_1}^T \cap {F_2}^2 \cap \cdots \cap {F_k}^T)=(F_1 \cap {F_1}^T) \cap (F_2 \cap {F_2}^T) \cap \cdots \cap (F_{k} \cap {F_{k}}^T)$. Now $F_k^T\subseteq G=F_1 \cap F_2 \cap \cdots \cap F_k$. Also $F_k^T\cap F_k=\emptyset$ as $F_k$ has loops at all its vertices and hence $F_k\cup F_k^T$ is complete. So $F_k^T\subseteq F_1 \cap F_2 \cap \cdots \cap F_{k-1}$ which implies $F_k\subseteq {F_1}^T \cap {F_2}^2 \cap \cdots \cap {F_{k-1}}^T$. Thus $G=(F_1 \cap {F_1}^T) \cap (F_2 \cap {F_2}^T) \cap \cdots \cap (F_{k-1} \cap {F_{k-1}}^T)=I_1 \cap I_2 \cap \cdots \cap I_{k-1}$ where $I_i=F_i\cap F_i^T$ for $i=1,2,\ldots ,k-1$. Since each $F_i$ has loop at every vertex, we have each $I_i$ is an interval graph by Observation \ref{obs:alpha}. Therefore $\BOX{G}\leqslant k-1$.

\vspace{1em} This limit is reached for $G=C_6$ (the cycle of length 6). Since $C_6$ is not an interval graph, but it can be easily obtained as an intersection graph of $2$-dimensional boxes, we have $\BOX{C_6}=2$. Now $D(C_6)=F_1\cap F_2\cap F_3$ where $F_i,\ i=1,2,3$, are Ferrers digraphs as represented below:

\vspace{1em}\noindent {\small $$\begin{array}{c|cccccc}
\multicolumn{1}{c}{} & a & b & c & d & e & f \\ \cline{2-7}
a & 1 & 1 & 0 & 0 & 0 & 1 \\
b & 1 & 1 & 1 & 0 & 0 & 0 \\
c & 0 & 1 & 1 & 1 & 0 & 0 \\
d & 0 & 0 & 1 & 1 & 1 & 0 \\
e & 0 & 0 & 0 & 1 & 1 & 1 \\
f & 1 & 0 & 0 & 0 & 1 & 1 
\end{array} = 
\begin{array}{c|cccccc}
\multicolumn{1}{c}{} & a & b & f & c & e & d \\ \cline{2-7}
a & 1 & 1 & 1 & 0 & 0 & 0 \\
b & 1 & 1 & 1 & 1 & 0 & 0 \\
f & 1 & 1 & 1 & 1 & 1 & 0 \\
c & 1 & 1 & 1 & 1 & 1 & 1 \\
e & 1 & 1 & 1 & 1 & 1 & 1 \\
d & 1 & 1 & 1 & 1 & 1 & 1 
\end{array} \bigcap 
\begin{array}{c|cccccc}
\multicolumn{1}{c}{} & c & d & b & e & a & f \\ \cline{2-7}
c & 1 & 1 & 1 & 0 & 0 & 0 \\
d & 1 & 1 & 1 & 1 & 0 & 0 \\
b & 1 & 1 & 1 & 1 & 1 & 0 \\
e & 1 & 1 & 1 & 1 & 1 & 1 \\
a & 1 & 1 & 1 & 1 & 1 & 1 \\
f & 1 & 1 & 1 & 1 & 1 & 1\end{array} \bigcap
\begin{array}{c|cccccc}
\multicolumn{1}{c}{} & b & c & a & d & f & e \\ \cline{2-7}
b & 1 & 1 & 1 & 1 & 1 & 1 \\
c & 1 & 1 & 1 & 1 & 1 & 1 \\
a & 1 & 1 & 1 & 1 & 1 & 1 \\
d & 0 & 1 & 1 & 1 & 1 & 1 \\
f & 0 & 0 & 1 & 1 & 1 & 1 \\
e & 0 & 0 & 0 & 1 & 1 & 1 
\end{array}$$}

\noindent So $d_F(D(C_6))\leqslant 3$. Again $H(M)=K_3$ where $M$ the following submatrix of $A(D(C_6))$:
$$\begin{array}{cc}
\begin{array}[b]{cc|ccc}
\multicolumn{2}{c}{} & b & f & d \\ \cline{3-5}
 & a & 1 & 1 & 0 \\
M\ = & c & 1 & 0 & 1 \\
 & e & 0 & 1 & 1  
\end{array} & \hspace{1.3in}\includegraphics[scale=0.4]{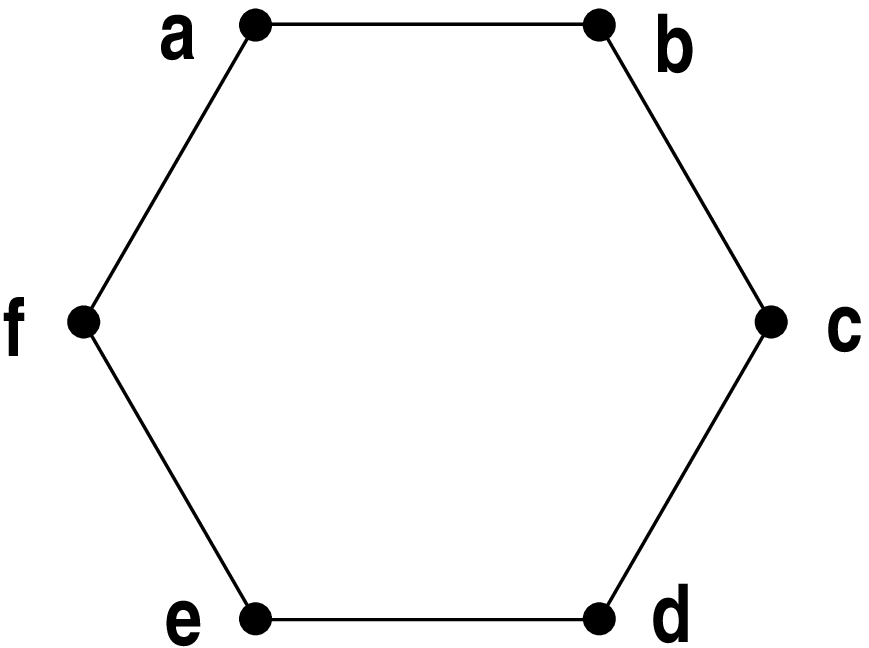}
\end{array}$$

\noindent Therefore $d_F(D(C_6))\geqslant 3$ and hence $d_F(D(C_6))=3$, as required.
\end{proof}

\section{A construction to determine the Ferrers dimension of a directed graph}

Let $D$ be a digraph and $H(D)$ be the graph associated to $D$. The following example shows that not every color class in a given coloring of $H(D)$ forms a Ferrers digraph.

\begin{exmp}
Let us consider the digraph $D$ whose adjacency matrix $A(D)$ is given below:
$$\begin{array}{c|ccc}
\multicolumn{1}{c}{} & a & b & c \\ \cline{2-4}
a & 1 & 0 & 0 \\
b & 0 & 1 & 0 \\
c & 0 & 0 & 1
\end{array}$$
The associated graph $H(D)$ is given by:
\begin{figure}[h]
\begin{center}
\includegraphics*[scale=0.5]{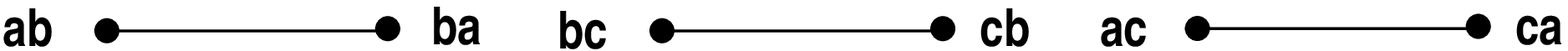}
\end{center}
\end{figure}

\noindent Consider the $2$-coloring of $H(D)$ with color classes $\set{ab,bc,ca}$ and $\set{ba,cb,ac}$. Now zeros of $A(D)$ corresponding to the color class $\set{ab,bc,ca}$,  forms the digraph whose adjacency matrix,
$$\begin{array}{c|ccc}
\multicolumn{1}{c}{} & a & b & c \\ \cline{2-4}
a & 1 & 0 & 1 \\
b & 1 & 1 & 0 \\
c & 0 & 1 & 1
\end{array}$$
shows that it is not a Ferrers digraph.
\end{exmp}

Thus it is clear that if a color class has to correspond a Ferrers digraph, it must contain all the zeros, which, among themselves, ensure the absence of couples. More precisely, if zeros $ab$ and $cd$ ($ab \neq cd$) are in the same color class, either $ad$ or $cb$ or both must also be in that same color class. In view of this observation, we modify the construction of Cogis and introduce the directed graph $J(D)$ instead of the undirected graph $H(D)$ corresponding to a digraph $D$.

\begin{defn}Let $D = (V, E)$ be a digraph. We define a digraph $J(D)$ with vertex set $E$ (i.e. the arcs of $D$) and there is an arc from $ab \in E$ to $cd \in E$ if and only if $ab \neq cd$ and $ad \in E$. \footnote{$a,b,c,d$ may not be all distinct.}
\end{defn}

It is clear that for any subdigraph $H$ of $D$, $J(H)$ is also a subdigraph of $J(D)$. Also $J(H)$ becomes an induced one whenever $H$ is an induced subdigraph of $D$.

\begin{exmp}Let us consider the following digraph $D$. Then the corresponding $J(D)$ is obtained as follows:
\begin{figure}[h]
\begin{center}
\includegraphics*[scale=0.4]{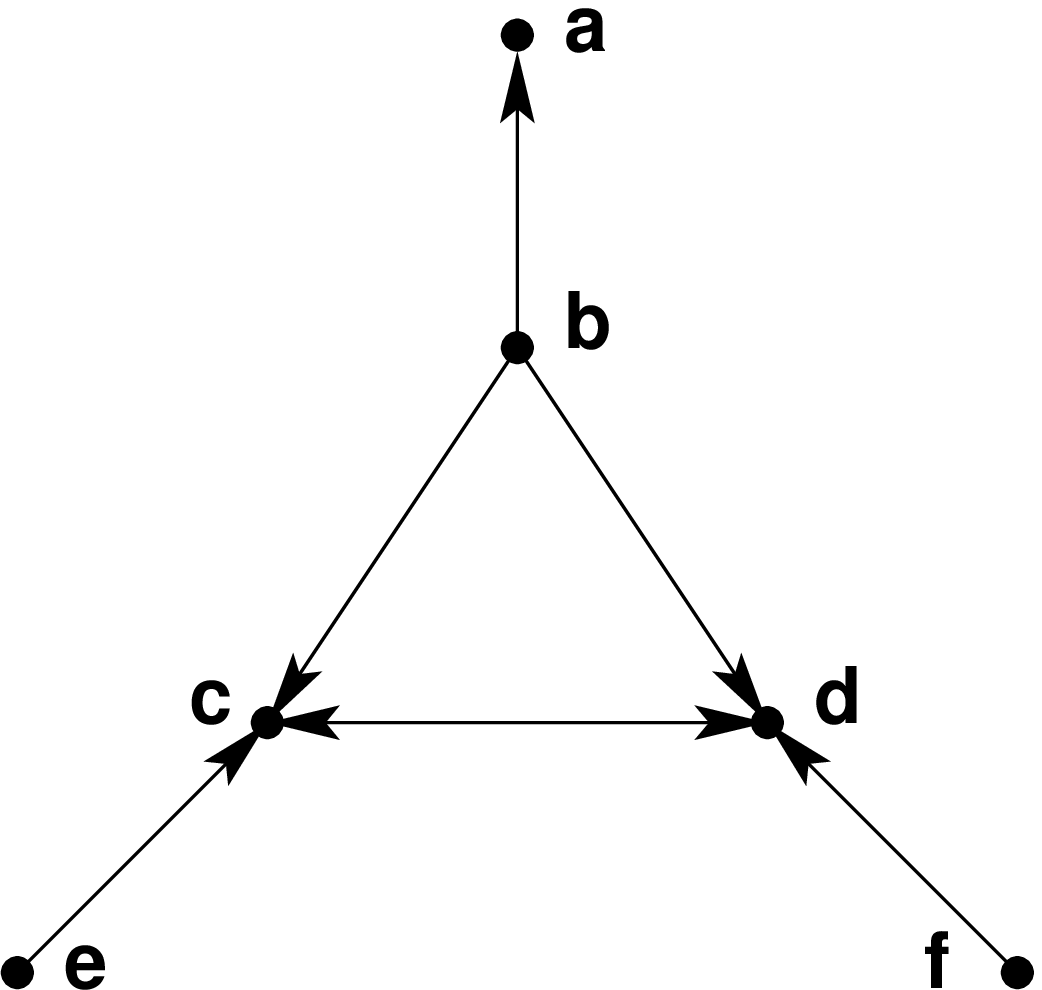}\hspace{1in}
\includegraphics*[scale=0.4]{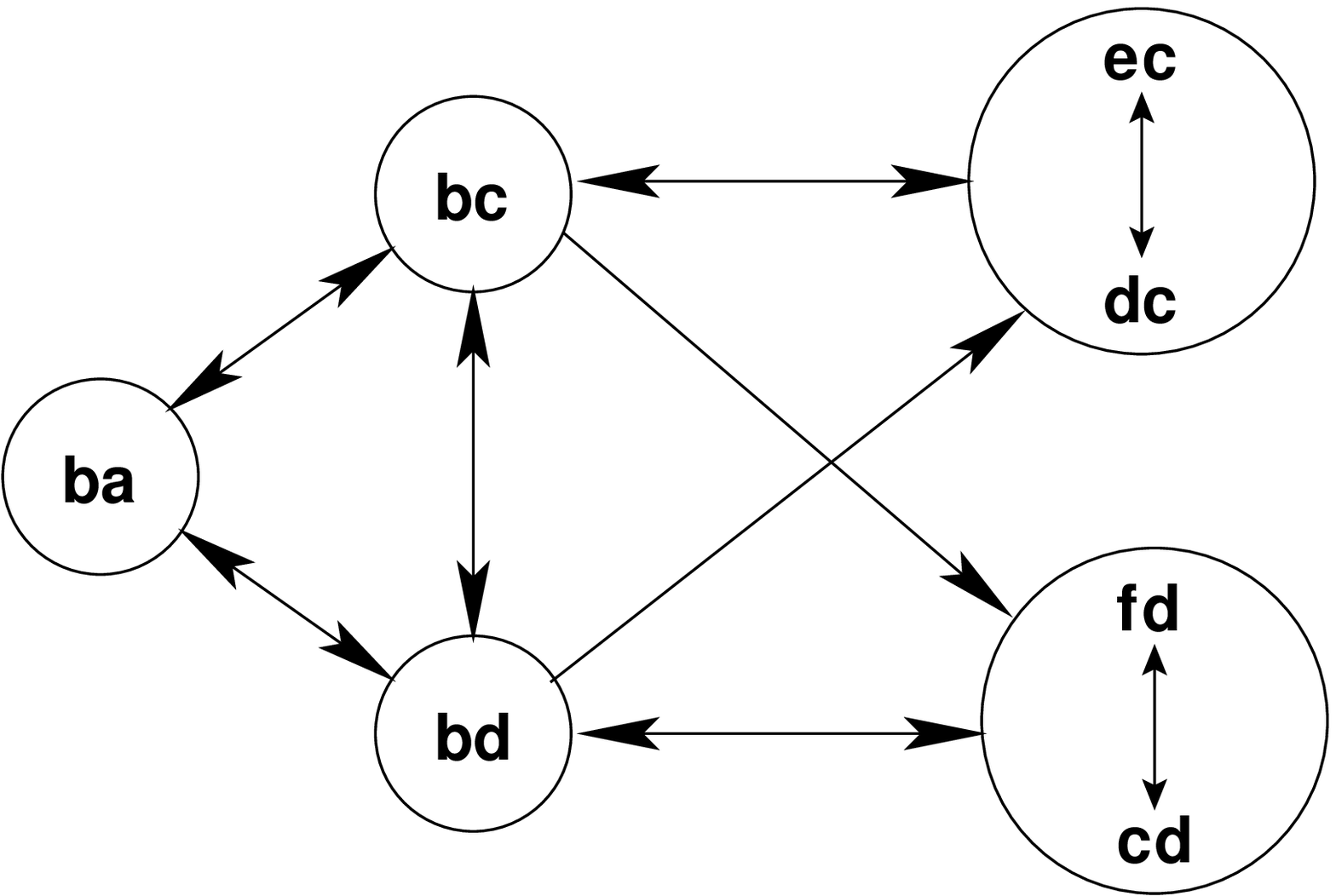}

$D$ \hspace{3in} $J(D)$
\end{center}
\end{figure}
\end{exmp}

\vspace{-2em} Certainly not all induced subdigraphs of $J(D)$ are of the form $J(H)$ for some subdigraph $H$ of $D$. For example, the induced subdigraph of $J(D)$ with the vertex set $\set{bc, ec, fd}$ is not of the form $J(H)$ for any subdigraph $H$ of $D$.

\begin{defn} A subdigraph $S$ of $J(D)$ with vertex set $V(S)$ is called an {\em ideal} subdigraph if
$$ab \longrightarrow cd \textrm{ in } S\ \Longrightarrow\ ab \neq cd\ \textrm{ and }\ ad \in V(S).$$
We note that for any subdigraph $H$ of $D$, $J(H)$ is an ideal subdigraph of $J(D)$.
\end{defn}

\begin{defn} An ideal subdigraph $S$ of $J(D)$ is called {\em total} if for any $ab\neq cd$ in $V(S)$, we have $ab\longrightarrow cd$ or $cd \longrightarrow ab$ or both (i.e., $ab\longleftrightarrow cd$). Let $D=(V,E)$. Then the {\em total covering number} of $J(D)$ is the minimum number of total subdigraphs of $J(D)$ needed to cover $E$, i.e., the vertex set of $J(D)$.
\end{defn}

For a digraph $D$, it is not known (as we mentioned above) that whether $d_F(D)=\chi(H(D))$, i.e., the clique covering number of $\overline{H(D)}$. But we have the following result:
\begin{thm} \label{thm:beta2}
Let $D$ be a digraph. Then the Ferrers dimension of $D$ is equal to the total covering number of $J(\overline D)$.
\end{thm}

\begin{proof}
Let $d_F(D)=n$ and the total covering number of $J(\overline D)$ be $m$. Then $D = F_1 \cap F_2 \cap \cdots \cap F_n$ for some Ferrers digraphs $F_i$, $i = 1, 2, \ldots , n$. Hence $\overline D = \overline {F_1} \cup \overline {F_2} \cup \cdots \cup \overline {F_n}$. Now we consider the subdigraph $\overline {F_1}$ of $\overline D$. We have, $J(\overline {F_1})$ is an ideal subdigraph of $J(\overline D)$ as $\overline{F_1}$ is a subdigraph of $\overline D$. We claim that $J(\overline {F_1})$ is a total subdigraph of $J(\overline D)$, whence it will follow that $m \leq n$.

\vspace{1em}Let $ab, cd \in E(\overline {F_1})$, $ab \neq cd$. Then $ab, cd \notin E(F_1)$ and so there are zero entries in the positions $ab$ and $cd$ the adjacency matrix of $F_1$. We have the following three cases:
$$\begin{array}{c|cc}
\multicolumn{1}{c}{} & b & d \\ \cline{2-3}
a & 0  &   \\
c &    &  0
\end{array} \hspace{1in}
\begin{array}{c|cc}
\multicolumn{1}{c}{} & b & d \\ \cline{2-3}
c=a & 0  &  0
\end{array} \hspace{1in}
\begin{array}{c|c}
\multicolumn{1}{c}{} & b = d \\ \cline{2-2}
a & 0 \\
c & 0
\end{array}$$
$$a \neq c, b \neq d
\hspace{1in} a = c, b \neq d
\hspace{1in} a \neq c, b = d$$
For the last two cases, $ad=cb=0$ and in the first case, since $F_1$ is a Ferrers digraph, $ad$ or $cb$ (or both) must be equal to zero. Thus $ad$ or $cb$ (or both) $\in V(J(\overline{F_1}))$, which implies $ab \rightarrow cd$ or $cd \rightarrow ab$ (or both) in $J(\overline {F_1})$. Therefore $J(\overline {F_1})$ is a total subdigraph of $J(\overline D)$ and the claim is verified.

\vspace{1em}Next, let $\left\{ S_1, S_2, \ldots, S_m\right\}$ be a total covering of $J(\overline D)$. For each $i = 1, 2, \ldots, m$, we define the subdigraph $F_i$ of $\overline D$ with the vertex set same as that of $D$ and edges which are belonging to the vertex set of $S_i$, i. e., $F_i = (V, V(S_i))$, where $D = (V, E)$. We show that $F_i$ is a Ferrers digraph by the method of contradiction. We assume that $F_i$ is not a Ferrers digraph so that there is a couple
$$\begin{array}{c|cc}
\multicolumn{1}{c}{} & b & d \\ \cline{2-3}
a & 1 & 0 \\
c & 0 & 1
\end{array}$$
in the adjacency matrix of $F_i$. Then $ab, cd \in E(F_i) = V(S_i)$, where $ab \neq cd$. Since $S_i$ is total, we have, $ad$ or $cb$ (or both) must belong to $V(S_i) = E(F_i)$ as $ab \rightarrow cd$ or $cd \rightarrow ab$ in $S_i$. This contradiction proves our assertion. Also since this covering covers all vertices of $J(\overline D)$, i.e., all the edges of $\overline D$, we have $\overline D = F_1 \cup F_2 \cup \cdots \cup F_m$, where each $F_i$ is a Ferrers digraph, so that $D = \overline{F_1} \cap \overline{F_2} \cap \cdots \cap \overline{F_m}$. Finally, since the complement of a Ferrers digraph is again a Ferrers digraph, we have $n \leq m$. This completes the proof.
\end{proof} 

One final remark is that the undirected graph obtained from $J(\overline D)$ by ignoring the directions of the arcs is the same as the complement of the graph $H(D)$. Now the chromatic number of $H(D)$ is the clique covering number of $\overline {H(D)}$, which is less than or equal to the total covering number of $J(\overline D)$ as every total subdigraph of $J(\overline D)$, made undirected by ignoring the direction of the arcs, is a clique in $\overline{H(D)}$, but the converse may not be true.

\end{document}